\documentclass[a4paper,12pt,reqno,superscriptaddress,showkeys]{revtex4}

\usepackage[centertags]{amsmath}
\usepackage{amsfonts}
\usepackage{amssymb}
\usepackage{amsthm}
\usepackage{newlfont}
\usepackage{stmaryrd}
\usepackage{mathrsfs}
\usepackage{euscript}
\usepackage{graphicx}
\usepackage{color}

\theoremstyle{plain}
  \newtheorem{theorem}{Theorem}[section]
  
  \newtheorem{proposition}[theorem]{Proposition}
  \newtheorem{lemma}[theorem]{Lemma}
\theoremstyle{definition}

\theoremstyle{remark}

\numberwithin{equation}{section}

\newcommand{\re}{\Re\,}

  \let\de=\delta 
   
 \let\la=\lambda

\newcommand{\caI}{{\mathcal I}}

\newcommand{\opunit}{\text{1}\kern-0.22em\text{l}}

\DeclareMathAlphabet{\mathpzc}{OT1}{pzc}{m}{it}

\newcommand{\id}{\textrm{d}}

\begin{document}

\title{{\Large{\bf Monotonicity of the dynamical activity}}}

\author{Christian Maes}
\affiliation{Instituut voor Theoretische Fysica, KU Leuven,
Belgium} \email{christian.maes@fys.kuleuven.be}
\author{Karel Neto\v{c}n\'{y}}
\affiliation{Institute of Physics AS CR, Prague, Czech Republic}
%\email{netocny@fzu.cz}
\author{Bram Wynants}
\affiliation{Institut de Physique Th\'eorique, CEA-Saclay, France}

\keywords{dynamical large deviations, nonequilibrium relaxation;\\
PACS numbers 02.50.Ga, 05.20.Dd, 05.40.-a, 51.10.+y}
%{\bf MSC2010}: 82C05, 60F10}

% --------------------------------------------------------------
\begin{abstract}
The Donsker-Varadhan rate function for occupation-time fluctuations has been seen numerically to exhibit monotone return to stationary nonequilibrium [Phys. Rev. Lett. 107, 010601 (2011)]. That rate function is related to dynamical activity and, except under detailed balance, it does not derive from the relative entropy for which the monotonicity in time is well understood. We give a rigorous argument that the Donsker-Varadhan function is indeed monotone under the Markov evolution at large enough times with respect to the relaxation time, provided that a ``normal linear-response'' condition is satisfied.
\end{abstract}

\maketitle
\section{Physical motivation and main finding}\label{pmo}
Equilibrium under given thermodynamic conditions is characterized by minimizing the appropriate thermodynamic potential.  For example, for a given temperature and pressure of the environment, an open system finds its equilibrium condition from minimizing the Gibbs free energy.  That is why water has to boil at around 100$^o$ Celsius under atmospheric pressure; that is the temperature above which the gas-phase of water gets a smaller Gibbs free energy.  That minimization is also dynamically realized.  Indeed, in many cases the second law of thermodynamics can be extended to become an H-theorem for the relaxation to equilibrium where the corresponding thermodynamic potential shows monotone decay to its equilibrium value.  Close-to-equilibrium the situation then resembles
the motion of a particle in a parabolic well with statistical forces derived from gradients in the equilibrium thermodynamic potential. That is indeed the framework of linear irreversible thermodynamics and the monotonicity in the approach to equilibrium relates essentially to the stability of the equilibrium phase.

The situation becomes much less clear when going to strong nonequilibrium regimes.  Nature is full of
stable nonequilibria for open systems in contact with different reservoirs, and yet no description in terms of thermodynamic potentials derived from energy and entropy (possibly supplemented with few other macroscopic quantities) has been found.  In fact, oscillations of the macroscopic condition are very well possible and can be effectively described by various models of dynamical systems.  But even restricting ourselves to a stationary macroscopic condition, or to a stationary open system in contact with large reservoirs, no \emph{physically} clear monotonicity has been observed in the approach to nonequilibrium. The numerical results observed in \cite{newp} open a new route that needs
further study.

The present paper addresses that general theoretical question, restricting to the simplest context of
Markov evolutions as described by Master or Fokker-Planck equations.  It is true that there the relative entropy with respect to the stationary distribution is always monotone.  Yet, that Lyapunov function is both well studied and not found directly related to known physical properties such as heat or work to which it relates under detailed balance.  This paper looks elsewhere and we find a natural candidate for monotonicity to be the Donsker-Varadhan functional governing stationary dynamical fluctuations.  It is physically related to the notion of dynamical activity, as will be explained below, a concept that has grown in importance for the elucidation of relaxation behavior in kinematically constrained systems.  In \cite{newp} we found that this functional is monotone in many examples.  The present paper adds a mathematical proof but we need an extra assumption: we show that this Donsker-Varadhan functional is monotone in the approach to stationary nonequilibrium under normal response behavior which is precisely stated as the sufficient condition in our main Theorem \ref{the}. We also add various examples that clarify the nature of that sufficient condition. As a result, the present paper is more technical than \cite{newp} and we now enter into more specific details.\\

Large deviation theory for Markov processes was developed by Donsker and Varadhan in 1975, \cite{DV}.  We recall the main setting.\\
Consider an ergodic Markov jump process $P_\rho$ with stationary probability law $\rho$ and with transition rates $k(x,y)$ over a finite state space $K$.  The empirical fraction of time that the system spends in state $x\in K$ over time-interval $[0,T]$ is
\begin{equation}\label{ot}
p_T(\omega,x) := \frac 1{T}\int_0^T \delta_{x_s,x}\,\id s,\quad x\in K
 \end{equation}
 where $\omega = (x_s, 0\leq s<T)$ is the piecewise constant trajectory.
By the assumed ergodicity, $p_T
\rightarrow \rho$ for $ T \uparrow +\infty$, $P_\rho-$almost surely. For the fluctuations
around that law of large times, there is a principle of large deviations, abbreviated as
\begin{equation}\label{abb}
P_\rho[p_T \simeq \mu] \propto e^{-T {\cal I}(\mu)},\quad T \uparrow +\infty
\end{equation}
for probability distributions $\mu$ on $K$,
in the usual logarithmic and asymptotic sense $T\uparrow +\infty$; see e.g. \cite{DZ,HS}.  %It implies that for all functions $f$
Equivalently, for all continuous functions $f$,
\[
\lim_{T\uparrow +\infty} \frac 1{T} \log
\Bigl\langle e^{\int_0^T f(x_t)\,\id t} \Bigr\rangle_\rho =
\inf_\mu \Bigl( \sum_xf(x)\mu(x) - {\cal I}(\mu) \Bigr)
\]
under the stationary expectation $\langle\cdot\rangle_\rho$.
\begin{equation}\label{dv}
{\cal I}(\mu) = \sup_{g>0}
\Bigl(
-\sum_x \frac{\mu(x)}{g(x)}\sum_y k(x,y)\, [g(y) - g(x)]
\Bigr)
\end{equation}
over positive functions $g$.

Since the functional ${\cal I}$ is strictly convex with unique minimum reached at the stationary distribution it is mathematically
natural to ask whether ${\cal I}(\mu_t)$ is also monotonically decaying to zero under the time evolution given by the Master equation
\[
\frac{\id}{\id t}\,\mu_t(x)  =
\sum_{y \in K} [k(y,x)\,\mu_t(y) - k(x,y)\,\mu_t(x)], \quad \mu_0=\mu
\]
at least when $\mu$ (the initial condition) is sufficiently close to $\rho$.  We found numerically in many cases that the answer is \emph{yes}, where the close-to-stationarity is verified for large times $t$. Then indeed $\mu_t$ gets sufficiently close to the stationary $\rho$.  That was reported in \cite{newp}.  A mathematical proof of monotonicity of the functional $\caI(\mu_t)$ for large times $t$ is lacking and the present paper will still need an additional assumption, called ``normal linear response.''
The normal linear response refers physically to the monotone decay of the linear response function, and mathematically it can be phrased as a sector condition on the backward generator $L$. The latter essentially means that the eigenvalues of $L$ should be contained in a wedge of the complex plane with a sufficiently small angle.  In particular it will be easy to show that the monotonicity holds when the system satisfies the condition of detailed balance, and hence, by a continuity argument, the monotone return to steady nonequilibrium is also valid {\it around} detailed balance. Yet again, as we will see below, the monotonicity often continues to hold even beyond the linear regime around detailed balance.\\

We now turn to the more physical motivation.
Recently, from the point of view of nonequilibrium statistical mechanics, there has been great interest in dynamical fluctuation theory, and in the occupation statistics in particular. The Donsker-Varadhan functional is exactly governing these fluctuations as mentioned above.  In that way our result mirrors the monotone behavior of the relative entropy which is associated to the {\it static} fluctuations of the system, \cite{sta}.  Apart from this more abstract analogy there are specific places where the functional ${\cal I}$ has turned up in nonequilibrium considerations. We know for example that close-to-equilibrium $\cal I$ is proportional to the excess in expected entropy production rate with respect to the stationary entropy production rate, \cite{minep}, which gives a fluctuation-based understanding of the minimum entropy production principle.  Nonperturbatively, the functional ${\cal I}(\mu)$ is an excess in expected dynamical activity (DA) as we now explain.

Usually for jump processes, one calls dynamical activity the quantity defined on path-space that counts the number of jumps or transitions.
Fixing any two distinct states $x, y$ the expected rate of jumping $x\rightarrow y$ when in $x$ is of course $k(x,y)$.  Therefore, under distribution $\mu$, the expected number of jumps per unit time is given by
\[
\xi(\mu) :=
\sum_x \mu(x) \sum_{y} k(x,y) =
\frac 1{2} \sum_{x, y} [\mu(x) k(x,y) + \mu(y) k(y,x)]
\]
which is the symmetric counterpart of the current (formally with $k(x,x) \equiv 0$). Since no confusion arises here, in the sequel we use the term \emph{dynamical activity} for the functional $\xi(\mu)$ evaluating the expected value of the more commonly defined variable dynamical activity on path-space.
The Donsker-Varadhan functional is the difference ${\cal I}(\mu) = \xi(\mu) - \xi_V(\mu)$ where $\xi_V$ is defined like $\xi$ but for modified rates $k_V(x,y) := k(x,y) \exp \{[V(y) - V(x)]/2\}$ with a potential $V$ so that the jump process with rates $k_V(x,y)$ makes $\mu$ stationary.  That will be explained in more detail in the next section,  and will be made most explicit in formula \eqref{das}.  The ${\cal I}(\mu)$ is thus an excess in expected activity between the original dynamics and a modified dynamics for which $\mu$ is made stationary.

On a broader level, the functional ${\cal I}(\mu)$ refers to a combination of properties of a statistical mechanical system that are related to its
reactivity and the ability to escape from its present state.  DA, and more specifically its version as defined on path-space, has been studied in connection with glassy behavior and the glass transitions; kinetically constrained models show a reduced dynamical activity over an extensive number of states which leads to dynamical phase transitions, \cite{vW,soll,chan,bodi}.  Finally, DA has appeared in fluctuation and response theory for steady nonequilibria, \cite{fdr,mprf,sas}.  The point is that as a function on trajectories, the dynamical activity is time-symmetric
and complements time-antisymmetric entropy fluxes whenever beyond the linear regime around equilibrium. That is why it enters the nonequilibrium fluctuation structure as well as provides extra contributions to the fluctuation-response relations.\\

Concerning physical implications of the observed (and here partially proven) time-monotonicity of the Donsker-Varadhan functional, we should again compare  with the situation of relaxation to equilibrium. There the existence of a physically meaningful (``Lyapunov'') functional, which does not increase over time, remarkably restricts the collection of admissible relaxation processes. Here we wish to proceed similarly to the strategy in equilibrium but we argue that out of equilibrium it may be useful to start from the DA as a fundamental quantity instead of from the entropy. The present paper contains no final judgement on this proposal but only a mathematically rigorous argument that this remains a valid possibility. It is proven that there is indeed a general tendency to decrease the excess dynamical activity analogous to the arrow of time associated with the increase of entropy under equilibrium condition. This is also related to the largely open problem of nonequilibrium statistical forces: the monotonicity of the Donsker-Varadhan functional suggests that its gradient with respect to macroscopic parameters can play the role of nonequilibrium statistical forces mimicking Onsager's theory of hydrodynamic entropy production. For some related physical arguments supporting the fundamental role of noise and dynamical activity out of equilibrium see e.g.~\cite{lan}.

The next section specifies the mathematical set-up and the main definitions.
Section \ref{mare} collects the main properties, with our result on monotone behavior. Section \ref{exs} discusses various specific examples far and close-to-equilibrium. 
Proofs are collected in Section \ref{pro} after which a final conclusion follows.

\section{Set-up}
As in the previous section, we consider a Markov jump process on a finite state space $K$ with states $x, y,\ldots$ and transition rates $k(x,y)$.  Probability distributions on $K$ will be denoted by $\rho, \mu, \nu,\ldots$.  The backward generator on functions $f$ is
\[
Lf (x) := \sum_{y \in K} k(x,y)\,[f(y) - f(x)]
\]
and its transpose generates the Master equation
\begin{equation}\label{jmu}
\frac{\id}{\id t}\mu_t(x)  +
\sum_{y \in K} j_{\mu_t}(x,y) = 0,\quad j_\nu(x,y) := k(x,y)\,\nu(x) - k(y,x)\,\nu(y)
\end{equation}
for the evolution on probabilities $\mu_t$ starting from some initial $\mu_0=\mu$ on $K$.  We assume that the Markov process is irreducible with unique stationary probability distribution $\rho$, i.e., $\rho(x) > 0$ solves $\sum_y j_\rho(x,y) = 0$ for all $x\in K$.

We say that the dynamics satisfies detailed balance when there is a function $U$ on $K$ for which
\begin{equation}\label{db}
k_e(x,y) \,e^{-U(x)} = k_e(y,x)\,e^{-U(y)}, \quad \rho_e(x)\propto e^{-U(x)}
\end{equation}
Here and below we decorate the rates and the stationary law in that detailed balance case with the subscript `e'.
Then, the free energy functional
\[
{\cal F}(\mu) := \sum_x \mu(x)\, U(x) + \sum_x \mu(x) \log \mu(x) \geq {\cal F}(\rho_e) = - \log \sum_x \exp\,[-U(x)]
\]
satisfies the monotonicity ${\cal F}(\mu_t)\downarrow {\cal F}(\rho_e)$ as a function of time $t$.
That is just a standard consequence of the general monotonicity of the relative entropy under stochastic transformations.  However, the relation between the Shannon entropy
$-\sum_x \mu(x) \log \mu(x)$ and physical notions as work or heat is mostly lost when far away from detailed balance.  A physically relevant alternative when moving away from detailed balance, is to consider the instantaneous entropy production ${\cal{E}}(\mu)$, which for the given context is
\begin{equation}\label{ep}
\begin{split}
{\cal{E}}(\mu) &:= \sum_{x,y}\mu(x)k(x,y)
\log\frac{\mu(x)k(x,y)}{\mu(y)k(y,x)}
\\
&= \frac{1}{2} \sum_{x,y} j_\mu(x,y) \,A_\mu(x,y),\qquad\qquad
A_\mu(x,y):= \log\frac{\mu(x)k(x,y)}{\mu(y)k(y,x)}
\end{split}
\end{equation}
as the  product of ``fluxes'' $j_\mu(x,y)$ and ``forces'' $A_\mu(x,y)$ when the system's distribution is $\mu$, reminiscent of irreversible thermodynamics --- see e.g. \cite{schnak} for more details.  

We now introduce our main object.
We embed the original dynamics into a larger family of processes with transition rates,
\begin{equation}\label{av}
k_W(x,y) := k(x,y)\,\exp \frac {W(y) - W(x)}{2}
\end{equation}
parameterized by functions $W$ on $K$. These functions $W$ are also called potentials.
Here we consider potentials that are directly connected with a probability distribution. What follows is a standard observation within the theory of large deviations, see e.g. Section 3.1.2 in \cite{DZ}, but for self-consistency
we give a full proof in Section \ref{sub:blowtorch}.

\begin{proposition}\label{bt}
For an arbitrary probability distribution $\mu > 0$ there exists a potential
$V = V_\mu$ on $K$ such that $\mu$ is invariant under the modified dynamics with transition rates $k_V(x,y)$.
The potential $V_\mu$ is unique up to an additive constant when the dynamics is irreducible.
\end{proposition}
In other words, for arbitrary $\mu >0$ we can always find a function $V$ so that
\begin{equation}\label{mv}
\sum_{y \in K}\big[ k_V(x,y)\,\mu(x) - k_V(y,x)\,\mu(y)\big] = 0, \quad x\in K
\end{equation}

We can compare this with \eqref{dv}.  Indeed, the Donsker--Varadhan large deviation functional can be written in terms of a potential $W$: taking $g=e^{W/2}$ in \eqref{dv},
\begin{equation}\label{das1}
  {\cal I}(\mu) = \sup_{W} \sum_{x,y \in K} \mu(x)\,[k(x,y) - k_W(x,y)]
\end{equation}
When the process is irreducible, cf. Proposition \ref{bt} and its proof in Section \ref{sub:blowtorch}, we then have
\begin{equation}\label{das}
{\cal I}(\mu) = \sum_{x,y \in K} \mu(x)\,[k(x,y) - k_V(x,y)]\quad \mbox{ with } \; V=V_\mu
\end{equation}
For physical motivation and as was mentioned already in Section \ref{pmo}, it is worth noting that ${\cal I}(\mu)$ is an excess or difference between the expected escape rates
$\sum_x\mu(x) \sum_y k(x,y)$ and $\sum_x\mu(x) \sum_y k_V(x,y)$.  Such an expected escape rate estimates the dynamical activity, i.e., the number of transitions per unit time in the process. We refer to the physics literature for further discussion, \cite{fdr,bodi,vW,soll,chan,sas}.

\section{Main result}\label{mare}

For simplicity in the sequel we always assume the irreducibility of continuous time Markov processes with finite state space.  Our main finding is that ${\cal I}(\mu_t)$ is monotone under the evolution~\eqref{jmu} when close enough to stationarity, i.e., for large enough times $t$ compared to the relaxation time, at least under some further and physically interpretable condition.\\

Define the real--space scalar product
$(f,g) := \sum_x f(x) g(x)\,\rho(x)$ so that $(f,Lg) = (L^*f,g)$.
% = (Lg,f)$.
Here, we have introduced the generator $L^*$ of the time-reversed process,
\[
L^*f(x) := \sum_{y \in K} \frac{\rho(y)\,k(y,x)}{\rho(x)}\,[f(y) - f(x)]
\]
We write $L_s$ for the symmetric part of the generator: $L_s := \frac 1{2}(L + L^*)$, and
\[
|||f||| := \max_{x,y}|f(x)-f(y)|
\]
for the variation of a function $f$ on $K$.  Now comes the main result of the paper.
\begin{theorem}\label{the}
Suppose that there is a constant $c>0$, so that $(L_sf,Lf) \geq c\,|||f|||^2$ for all functions $f$ on $K$.  Then, there is a time $t_o>0$ so that for all initial probability distributions $\mu$
on $K$,
\[
\frac{\id}{\id t}\, {\cal I}(\mu_t) \leq 0 \;\;\;\mbox{ for all times } t\geq t_o
\]
\end{theorem}
\smallskip
In section \ref{timeses} it will be shown that the time $t_o$ after which monotonicity sets in is of the order of the relaxation time (inverse of the exponential rate of convergence).

Since $L$ and $L_s$ have a bounded inverse on the functions $f$ that have zero mean $\sum_x \rho(x) f(x) =0$, the condition of Theorem~\ref{the} in essence means to require that $(L_s f, Lf) > 0$ for all non-constant $f$.
In fact, it is sufficient and more convenient to verify the inequality
%for $(L_sf,Lf) > 0$ is
\begin{equation}\label{sec}
(f,L^2f) > 0
\end{equation}
which is usually called (a specific instance of) a sector condition. In particular, for $L$ a normal operator this is equivalent to the assumption that all its non-zero eigenvalues
$\la = -a + i b \neq 0$ obey the inequality
$\re(\la^2) > 0$, i.e., $|b| < a$.
Obviously, the condition~\eqref{sec} is fulfilled whenever the rates $k(x,y)$ satisfy detailed balance~\eqref{db}, $L_e = L^*_e$, see under the section \ref{adb} for further discussion. By a continuity argument, this also extends to dynamics where the detailed balance is only weakly violated.\\

A physical interpretation of the hypothesis for Theorem \ref{the} is in terms of the generalized susceptibility for the linear response around the stationary probability $\rho$. For a function $B$ on $K$ consider perturbed transition rates
\begin{equation}\label{pert}
k(\tau;x,y) := k(x,y)\,e^{\frac{h_\tau}{2}[B(y) - B(x)]},\qquad
\tau\geq 0
\end{equation}
with small time-dependent amplitude $h_\tau$, $|h_\tau| \leq \varepsilon$.  It resembles \eqref{av} but now the perturbation is time-dependent.  That new time-dependent process is started from time zero in the distribution $\rho$ which is stationary for $h\equiv 0$.  At a later time $t>0$, in the process with rates \eqref{pert}, when taking the expectation of a function $G$, we see the difference
\[
\langle G(x_t) \rangle^h -  \sum_x \rho(x) G(x)  =
\int_0^t\id \tau\,h_{t-\tau} \,\chi_{GB}(\tau) + O(\varepsilon^2)
\]
which defines the generalized susceptibility $\chi_{GB}$.  There is an explicit formula extending the standard fluctuation--dissipation theorem, see e.g.~\cite{fdr}:
\begin{equation}
\chi_{GB}(t) := \frac{\delta}{\delta h_0}\langle G(x_t) \rangle^h\Bigr|_{h=0} = -\frac{1}{2}\Big[\frac{d}{dt}\Big<B(x_0)\,G(x_t)\Big>_{\rho} + \Big<LB(x_0)\,G(x_t)\Big>_{\rho} \Big]\label{linr}
\end{equation}
with right-hand expectations in the original stationary process $P_\rho$.
Our next result gives an explicit expression for the zero-time susceptibility in the case $G = B$; see Section~\ref{linresp-proof} for a proof.
%It is then a simple computation, that we do in~\eqref{core},
%to conclude that for $t\downarrow 0$,
\begin{proposition}\label{normal}
We have the identities
\begin{equation}\label{ff}
\frac{\id}{\id t}\,\chi_{ff}(t) \Bigr|_{t=0} = \chi_{Lf,f}(0) = - (L_sf,Lf)
\end{equation}
\end{proposition}

As a consequence, the hypothesis of Theorem~\ref{the} can be rephrased as the condition
$\chi_{ff}(t) \leq \chi_{ff}(0)$ for small enough $t > 0$. Such an equilibrium-like response behavior at initial times is called ``normal linear response'' throughout the paper.

\section{Examples}\label{exs}

We illustrate the statements of the previous section by providing some examples.

\subsection{Asymmetric diffusion on the ring}\label{asy}

Consider a ring consisting of $N > 2$ sites,
labeled $x = 1,2,\ldots , N +1\equiv 1 $. For a totally asymmetric random walker the only non-zero transition rates are of the form $k(x,x+1) > 0$ with
Master equation \eqref{jmu} simplified to
\[ \frac{\id \mu_t}{\id t}(x) = \mu_t(x-1)k(x-1,x) - \mu_t(x)k(x,x+1) \]
The stationary distribution is
\[
\rho(x) = \frac{C}{k(x,x+1)},\ \ \ \ \ C^{-1} = \sum_x\frac{1}{k(x,x+1)}\]
The corresponding generators of Theorem \ref{the} are
\begin{align*}
  Lf(x) &= k(x,x+1)[f(x+1) - f(x)]
\\
  L_sf(x) &= \frac 1{2}k(x,x+1)[f(x+1) + f(x-1) - 2f(x)]
\end{align*}
so that the hypothesis of Theorem \ref{the} concerns
\begin{equation}\label{condi}
(L_sf,Lf) = \frac{C^2}{2}\,\sum_x \frac 1{\rho(x)}\,[f(x+1) - f(x)]\,[f(x+1) + f(x-1) - 2f(x)]
\end{equation}

\emph{Homogeneous case.}
%To begin let us take
Let $k(x,x+1) \equiv p$ for some given $p>0$. In that case, \eqref{condi} is bounded from below by the variation of $f$ because, with $g(x)=f(x+1)-f(x)$, $\sum_x g(x) [g(x)- g(x-1)] = \frac 1{2}\sum_x [g(x)- g(x-1)]^2\geq 0$ on the ring.\\
But we can also explicitly show that ${\cal I}(\mu_t)$ is monotone for all times $t \geq 0$ and starting from all possible $\mu>0$.  For this we find the potential $V_\mu$ solving
\[ 0 = \mu(x)\,p\, e^{[V(x+1)-V(x)]/2} - \mu(x-1)\,p\,e^{[V(x)-V(x-1)]/2} \]
or
\[ \frac{V(x+1)-V(x)}{2} = -\log \mu(x) +  \frac 1{N}\sum_y\log\mu(y)
 \]
Therefore the Donsker-Varadhan functional \eqref{das} equals
\begin{eqnarray*}
{\cal I}(\mu) &=& \sum_x p \mu(x)\left[1 - e^{\frac{V(x+1) - V(x)}{2}}\right]\\
&=& p - pN\Bigl[\prod_{y=1}^N\mu(y)\Bigr]^{\frac{1}{N}}
\end{eqnarray*}
in terms of the geometric mean of the $\mu(1), \mu(2),\ldots, \mu(N)$.
The time derivative at $\mu_t=\mu$ is computed to be
\begin{eqnarray*}
\frac{\id}{\id t}{\cal I}\,(\mu_t)
&=&  - p^2 \Bigl[ \prod_{y=1}^N\mu(y) \Bigr]^{\frac{1}{N}}\sum_{x=1}^N
\Bigl( \frac{\mu(x-1)}{\mu(x)} - 1 \Bigr)
\end{eqnarray*}
which is non-positive by applying Jensen's inequality as
\[ \log\left(\frac{1}{N}\sum_{x=1}^N\frac{\mu(x-1)}{\mu(x)}\right) \geq \frac{1}{N}\sum_{x=1}^N\log \frac{\mu(x-1)}{\mu(x)} = 0 \]
In fact, by the same argument, the time-derivative is strictly negative whenever $\mu\neq \rho$.
Therefore, for homogeneous totally asymmetric walkers on a ring, we always have monotonicity of the geometric mean of the occupations $\mu_t(x)$, as a consequence of the monotonicity of the Donsker-Varadhan functional.\\

\emph{Inhomogeneous case.}
We look back at \eqref{condi} which is now of the form $\sum_x \frac 1{\rho(x)} g(x)[g(x)-g(x-1)]$,
always with $g(x) = f(x) - f(x-1)$.
An explicit computation of the time-derivative of $\caI(\mu_t)$ gives, with
$\mu_t = \mu$,
\begin{equation}\label{td}
\frac{\id}{\id t}\,{\cal I}(\mu_t)
= 2\epsilon^2C^2\sum_x\frac{V(x)}{\rho(x)}\,[V(x-1)-V(x)] + o(\epsilon^2)
\end{equation}
with $V(x) = -h(x) + \frac{1}{N}\sum_x h(x)$ for $\mu(x) =\rho(x)[1+ \epsilon\, h(x)]$.   We thus see how the condition in Theorem~\ref{the} appears.  Without further condition and depending on the shape of the stationary distribution $\rho$, this time-derivative \eqref{td} can be either positive or negative. For an example making ${\cal I}(\mu_t)$ non-montone at initial times, we take $N=4$ with rates $k(1,2)=30,  k(2,3)= k(3,4)= k(4,1)= 1$, and $\mu$ determined from $\epsilon = 0.02$ and
\begin{equation}\label{initial}
V(1) = 1, V(2) = -3, V(3) = 0, V(4) = 2
\end{equation}
Then the time-derivative \eqref{td} is positive whenever
\[ \frac{1}{\rho(1)} > \frac{12}{\rho(2)} + \frac{4}{\rho(4)} \]
To visualize this example we show in Fig.~\ref{fig:ce} the result of a numerical computation of ${\cal I}(\mu_t)$ for this initial condition (\ref{initial}).
\begin{figure}[ht]
\includegraphics[width=0.7\linewidth]{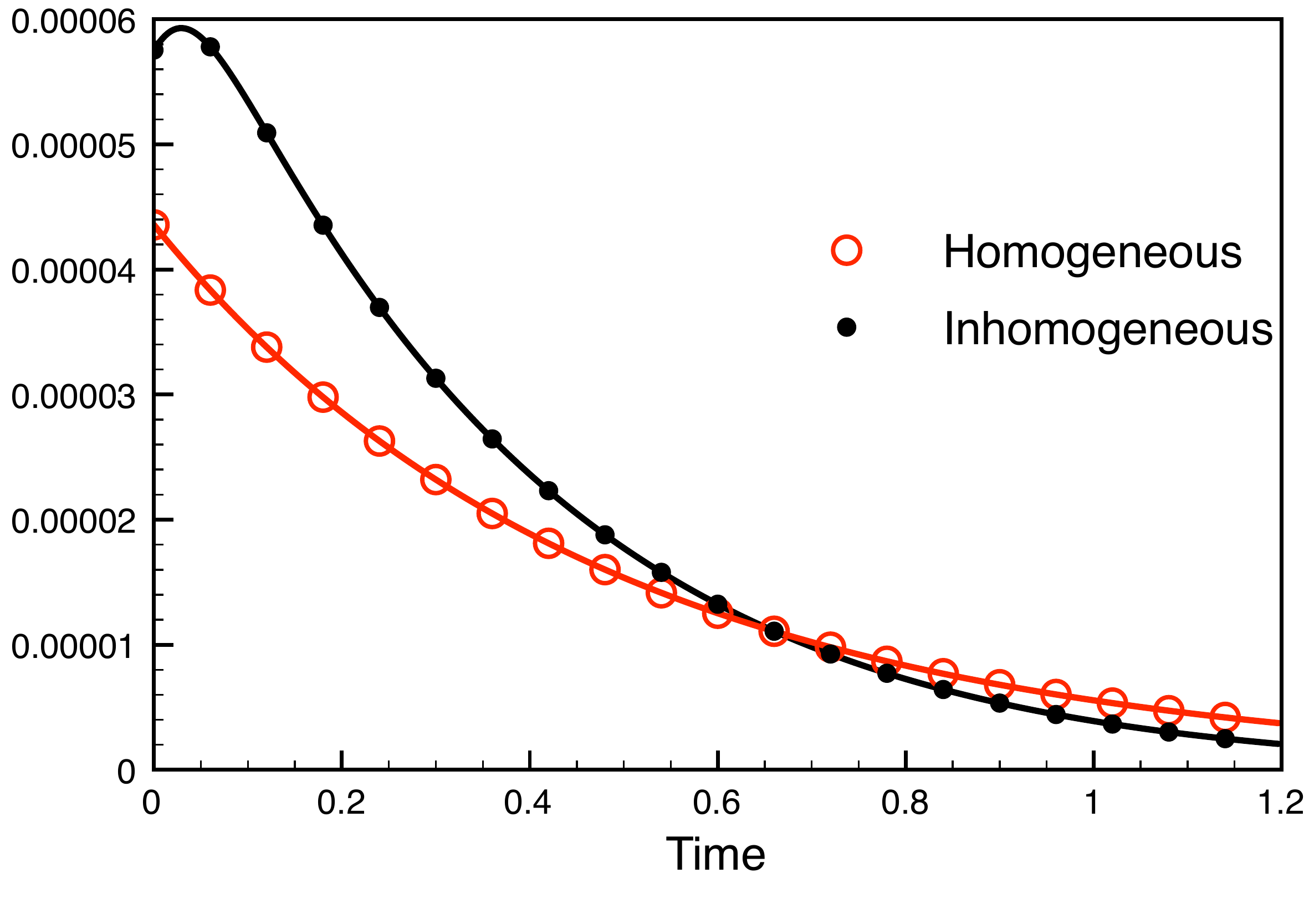}
\caption{The functional ${\cal I}(\mu_t)$ as a function of time. The curve with open circles shows the evolution of ${\cal I}(\mu_t)$ for the case of a homogeneous stationary distribution. The black curve (closed smaller circles) represent the case where the stationary distribution is inhomogeneous: $\rho(1) = 1/91, \rho(2) = \rho(3)=\rho(4) = 30/91$.}
\label{fig:ce}
\end{figure}
Observe however, even in this example, that after a short initial slip
${\cal I}(\mu_t)$ starts decreasing monotonically.  In other words, while for all $\epsilon>0$ there is a probability $\mu$ in the  neighborhood of the stationary law $\rho$ with variational distance $d(\mu,\rho) \leq \epsilon$, so that ${\cal I}(\mu_t)$ is not monotone at $\mu_t=\mu$, still ${\cal I}(\mu_t)$ decays monotonically to zero \emph{eventually} (after a long enough time $t$).  This also indicates that the hypothesis in Theorem \ref{the} is not at all necessary.\\

An example of a driven diffusion process on the ring was treated in \cite{newp}.  There,  in Fig.~1 of \cite{newp}, we have seen that ${\cal I}(\mu_t)$ keeps decaying monotonically while the entropy production ${\cal E}(\mu_t)$, the continuum version of \eqref{ep}, oscillates in time.

\subsection{Detailed balanced dynamics}\label{adb}

Under detailed balance \eqref{db}, we can take $g=\sqrt{\rho_e/\mu}$ in \eqref{dv} to find equality with the Dirichlet form
\[
{\cal I}_e(\mu) = -\sum_x\rho_e(x) \, \sqrt{\frac{\mu(x)}{\rho_e(x)}}\Bigl(L_e\sqrt{\frac{\mu}{\rho_e}}\Bigr)(x) = -\Bigl(\sqrt{\frac{\mu}{\rho_e}},L_e\sqrt{\frac{\mu}{\rho_e}}\Bigr)
\]
Furthermore, for $\mu_t=\mu$ at time $t$
\[ \dot{\mu}_t=   \rho_e\, L_e \bigl( \frac{\mu}{\rho_e} \bigr) \]
so that we get the time derivative
\begin{equation}\label{timede}
\frac{\id}{\id t}\,{\cal I}_e(\mu_t) =
- \sum_x \frac {\rho_e(x)}{f(x)}\,(L_e f^2)(x)\,(L_e f)(x)
\end{equation}
where we have abbreviated $f = \sqrt{\mu/\rho_e}$.   Obviously, that time-derivative is negative whenever $f-1$ is sufficiently small (close-to-stationarity): with $f(x)  = 1 + \epsilon\,h(x)$,
\[ \frac{\id}{\id t}\,{\cal I}_e(\mu_t) =
- 2\epsilon^2\sum_x \rho_e(x) \,[(L_e h)(x)]^2 + o(\epsilon^2)
\]
That is equivalent to what was mentioned before under Theorem \ref{the}: the hypothesis there is always satisfied under detailed balance, and we see that the time $t_0$ will be of the order of the relaxation time, characterizing the uniform exponentially fast convergence to equilibrium.  More explicit calculations reveal also that the derivative \eqref{timede} is strictly negative for all $f>0$ when $|K| \leq 3$.

%However, \eqref{timede} is not always negative.
On the other hand, we can check that the time-derivative~\eqref{timede} may obtain positive values too. For convenience, we briefly consider a one-dimensional diffusion variant of the same problem, with the forward generator
%the case of a diffusion process where the forward generator takes the form of the differential operator, say on $\bbR$,
\[
L_e g = \frac 1{\rho_e}\frac{\id}{\id x}\big(\rho_e\, \frac{\id g}{\id x} \big)
\]
for smooth functions $g$.
Let us take $\mu(x) = a/\rho_e(x)$ inside the interval $[1,\ell]$, where $\ell > 1$ and $a$ is a normalization, and $\mu(x)\simeq \rho_e(x)$ very rapidly decaying to zero outside that same interval.  Then, at that $\mu$, \eqref{timede} becomes
\[ \frac{\id}{\id t}\,{\cal I}_e(\mu_t) \simeq - 2\int_1^\ell\id x \,f''(x)\, (\log f)''(x) \]
which will be positive e.g.\ when $f$ is convex while $\log f$ is concave --- for example, with $\rho_e(x) = c_\ell/x^2$ on $[1,\ell]$, we have $f(x) \sim x^2$, $\log f(x) \sim \log x$ on that same interval.

\section{Proofs}\label{pro}

\subsection{Proof of Proposition~\ref{bt}}\label{sub:blowtorch}

For arbitrary $\mu > 0$ we consider the auxiliary functional
\begin{equation}
  Y_\mu(W) := \sum_{x,y \in K} \mu(x)\,k_W(x,y)
  = \sum_{x,y \in K} \mu(x)\,k(x,y)\, \exp \Bigl[
  \frac{W(y) - W(x)}{2} \Bigr]
\end{equation}
defined on all functions on $K$.  Of course, since the value only depends on the differences $W(y)-W(x)$ we can as well take $W \in C_0(K)$, the collection of all functions that are equal to zero on a fixed ``root'' $x_0 \in K$. This functional $Y_\mu$ is nonnegative and convex,
\begin{equation}\label{convex}
  Y_\mu(\la W_1 + (1-\la) W_2) \leq \la Y_\mu(W_1) + (1-\la) Y_\mu(W_2)
\end{equation}
by convexity of each contribution $\mu(x)\,k_W(x,y)$.
Below in Lemma~\ref{lem:technical} we prove that under the irreducibility assumption, $Y_\mu$ is actually \emph{strictly} convex and that it attains inside $C_0(K)$ a unique minimum at some $W = W_\mu$. Hence, $W_\mu$ is also a  minimizer (unique up to an additive constant) on the unconstrained space of all functions on $K$, implying that for all $x \in K$,
\begin{equation}
\begin{split}
  0 &= \frac{\de Y_\mu}{\de W(x)} \Bigr|_{W = W_\mu} =
  \frac{1}{2} \sum_{y \in K} [\mu(y)\,k_W(y,x) -
  \mu(x)\,k_W(x,y)]
\end{split}
\end{equation}
which is just the stationarity of $\mu$ for the dynamics with rates $k_W(x,y)$,
i.e., $V = V_\mu$ of Proposition~\ref{bt} does exist and equals $W_\mu$. This also proves formula~\eqref{das1} as
\begin{equation}
  {\cal I}(\mu) = \sum_{x,y \in K} \mu(x)\,k(x,y) - Y_\mu(W_\mu)
  = \sup_V \sum_{x,y \in K} \mu(x)\,[k(x,y) - k_V(x,y)]
\end{equation}

A general reducible dynamics can be decomposed into irreducible components (including isolated sites) and for each of them the above argument holds true, i.e.,
the supremum on the right-hand side of~\eqref{das1} is attained on
a function $V_\mu$, which is also a solution of the inverse stationarity problem and which is unique up to a constant within each component.\\

Again turning to irreducible dynamics, we  have
\begin{lemma}\label{lem:technical}
For any $\mu > 0$, $Y_\mu|_{C_0(K)}$ is strictly convex and has a unique minimum.
\end{lemma}
\begin{proof}
By irreducibility, there exists a cyclic sequence of states $(x_0,x_1,...,x_n = x_0)$ that covers the whole space $K$ and such that for all consecutive pairs of states,
$\mu(x_{i-1})\,k(x_{i-1},x_i) \geq \de$ with some $\de > 0$. If $W_1$ and $W_2$ are such that the relation~\eqref{convex} becomes an equality, then, for all $i = 1,\ldots,n$,
$W_1(x_i) - W_1(x_{i-1}) = W_2(x_i) - W_2(x_{i-1})$, by using that the exponential is strictly convex. From $W_1(x_0) = W_2(x_0) = 0$ then follows $W_1 = W_2$, identically. This proves the strict convexity and hence
the uniqueness of the minimum for $Y_\mu|_{C_0(K)}$.

To prove that the minimum exists, we consider the compact sets
\[
  C_0^a(K) := \{W \in C_0(K);\,|W(x)| \leq a\,\,\text{for all } x \},\quad a > 0
\]
and define $M_\mu := Y_\mu(0) = \sum_{x,y \in K} \mu(x)\,k(x,y)$. By construction, for any
$W \in C_0(K) \setminus C_0^a(K)$ there exists $i$ such that
$W(x_i) - W(x_{i-1}) > a/n$ and hence
$Y_\mu(W) > \de\,e^{a/(2n)}$. Fix now some $a$ so that $\de\,e^{a/(2n)} > M_\mu$. By compactness,
$Y_\mu$ on the set $C_0^a(K)$ attains the minimum, which then coincides with the minimum of
$Y_\mu|_{C_0(K)}$.
\end{proof}

\subsection{The map $\mu \mapsto V_\mu$}
Most importantly, from the previous section, the map $\mu \mapsto V_\mu$ is a bijection when we think of the potential modulo a constant.  Moreover $V_\mu$ depends smoothly on $\mu$, and {\it vice versa}. In other words, the map $\mu \mapsto V_\mu$ is a diffeomorphism with variational distance $\id(\mu,\rho)$ of the same order as $V$:
\[
c_0\,\id(\mu,\rho) \leq |||V_\mu||| \leq c_1\, \id(\mu,\rho)
\]
for constants $c_0,c_1 >0$. That is really a consequence of the irreducibility of the finite Markov process, or see chapter two in \cite{kato}.

Heuristically it suffices to understand the linearized map around $\mu=\rho, V_\rho=0$ since the modified rates $k_{V_1+V_2}(x,y)= k_{V_1}(x,y)\exp\{V_2(y)-V_2(x)\}/2$ each time define an irreducible Markov process for each $V_1$.  Writing $\mu =\rho(1+ \varepsilon h)$ for some function $h$ with mean $\sum_x \rho(x) h(x)=0$ and for small $\varepsilon$, we easily find $V_\mu = \varepsilon v + O(\varepsilon^2)$ with
\begin{equation}\label{sta}
L_s v = L^* h
\end{equation}
Note here that $\rho$ is also invariant under the time-reversed process and under the (detailed balanced) process generated by $L_s$.  Hence, $\sum_x \rho(x) L^*h(x) =0$ and $L^*h$ is in the domain of the (Drazin) pseudo-inverse $(L_s)^{-1}$, so that \eqref{sta} has a unique solution (again up to a constant); in fact $h=0$ if and only if $v=0$.\\

The computation leading to \eqref{sta} goes as follows. For all $x\in K$,
\[
0 = \sum_y \bigl( \mu(y)k(y,x)e^{\frac{V(x)-V(y)}{2}} - \mu(x)k(x,y)e^{\frac{V(y)-V(x)}{2}}\bigr)
\]
which by expanding the exponential directly yields the identity
\begin{equation}\label{we}
L^* \bigl(\frac{\mu}{\rho}-1 \bigr)(x) - L_sV_\mu(x)  = w(x,\mu)
\end{equation}
where (with $V=V_\mu$)
\begin{multline}\label{ide}
w(x,\mu) :=
\sum_y\bigl\{\bigl(\frac{\mu(y)}{\rho(y)}-1\bigr)\frac{\rho(y)k(y,x)}{\rho(x)}\frac{V(x)-V(y)}{2} + \bigl( \frac{\mu(x)}{\rho(x)}-1)k(x,y \bigr) \frac{V(x)-V(y)}{2} \bigr\}
\\
+\sum_y \big\{\frac{\mu(y)}{\rho(x)}k(y,x)-\frac{\mu(x)}{\rho(x)}k(x,y)\bigr\} \delta_V(x,y)
\end{multline}
for $\delta_V(x,y):= \sum_{n=2}[\frac{V(x)-V(y)}{2}]^n\frac 1{n!}$.  Each difference
\[
  \Bigl|\frac{\mu(x)}{\rho(x)}-1 \Bigr| \leq C_0\,|||V_\mu|||, \quad x\in K
\]
so that $|w(x,\mu)|\leq C_1\,|||V_\mu|||^2$ for some constant $C_1$ when $|||V_\mu|||$ is sufficiently small.

\subsection{Proof of Theorem \ref{the}}\label{timeses}
We must take the time-derivative of ${\cal I}(\mu_t)$,
\begin{equation}\label{timed}
\begin{split}
\frac{\id}{\id t}\,{\cal I}(\mu_t) &= -\frac 1{2}\sum_{x,y}\dot{\mu_t}(x) k(x,y)\big[V_t(y) -V_t(x) +2 \delta_V(y,x)\big]
\\
&\phantom{**}+ \frac 1{2}\sum_{x,y}\mu_t(x)k_{V_t}(x,y)\big[\dot{V_t}(x) - \dot{V_t}(y)]
\end{split}
\end{equation} where $\dot{V}_t(x) := \frac{\id}{\id t}V_{\mu_t}(x)$.  The second line in \eqref{timed} equals zero because per fixed time $t$, $\mu_t$ is stationary for the dynamics with rates $k_{V_t}(x,y)$.
We thus have
\[
\frac{\id}{\id t}\,{\cal I}(\mu_t) = -\frac 1{2}\sum_{x}\dot{\mu_t}(x)\big[LV_t (x) + 2\sum_y k(x,y) \delta_V(y,x)\big]
 \]
Looking at the first term, we use that
\[
\dot{\mu_t}(x) = \rho(x) \,L^*(\frac{\mu_t}{\rho}-1)(x) = \rho(x) L_sV_t(x) + \rho(x)w(x,\mu_t)
\]
as introduced in \eqref{we}.  In other words, we have obtained
\[
\frac{\id}{\id t}{\cal I}(\mu_t) = -\frac 1{2} Q(V_t) - \frac 1{2}\sum_x\rho(x)w(x,\mu_t)\,LV_t(x) - \sum_x\dot{\mu}_t(x)k(x,y)\delta_{V_t}(y,x)
\]
for the quadratic form $Q(f) := (Lf,L_sf)= (f,L^*L_s f) = (L_s Lf,f)$ which, from the hypothesis of Theorem \ref{the}, is bounded from below by $c\,|||f|||^2$.   Since $|||V_t||| \leq K\exp[-\gamma t]$ for some $K<\infty, \gamma >0$, it suffices finally to realize that, at least for large enough times $t$,
\[
\Bigl|\frac 1{2}\sum_x\rho(x)w(x,\mu_t)LV_t(x) + \sum_x\dot{\mu}_t(x)k(x,y)\delta_{V_t}(y,x) \Bigr| \leq
C \, |||V_t|||^{3}
 \]
for some $C<\infty$.  That easily follows by applying uniform bounds such as
$|LV(x)|\leq C_2\,|||V|||$ and $|\delta_V(x,y)| \leq C_3\,|||V|||^2$ for small enough $V$, combined with the previous estimate $|w(x,\mu)|\leq C_1\,|||V_\mu|||^2$ making also $|\dot{\mu_t}(x)|
\leq C_4\,|||V_t|||$.  That concludes the proof of Theorem \ref{the}.\\

The proof above obviously gives an estimate of the time $t_0$ after which monotonicity surely sets in.
Since $|||V_t|||$ is of the order $\exp(-\gamma t)$,\, $Q(V_t)>0$ dominates the time-derivative of the Donsker-Varadhan functional when $\gamma t \gg 1$, i.e., for times beyond the relaxation time to the nonequilibrium steady regime.

\subsection{Proof of Proposition~\ref{normal}}\label{linresp-proof}

A straightforward computation gives, for all $t \geq 0$,
\begin{equation}
\begin{split}
\chi_{ff}(t) &= -\frac{1}{2}\Big[\frac{d}{dt}\Big<f(x_0)f(x_t)\Big>_{\rho} + \Big<Lf(x_0)f(x_t)\Big>_{\rho} \Big]
\\
&= -\frac{1}{2}\Big<(L+L^*)f(x_0)f(x_t)\Big>_{\rho}
\end{split}
\end{equation}
and hence
\begin{equation}\label{core}
  \frac{d}{\id t}\,\chi_{ff}(t)\Bigr|_{t=0} =  -(L_sf, Lf) = -Q(f)
\end{equation}
Similarly, we look at the response of the observable $Lf$ to find
\begin{eqnarray}
\chi_{Lf,f}(t) &=& -\frac{1}{2}\Big[\frac{d}{dt}\Big<f(x_0)Lf(x_t)\Big>_{\rho} + \Big<Lf(x_0)Lf(x_t)\Big>_{\rho} \Big]\nonumber\\
          &=& -\frac{1}{2}\Big<(L^*+L)f(x_0)Lf(x_t)\Big>_{\rho}\nonumber\\
\chi_{Lf,f}(0) &=&  - (L_sf, Lf)\label{ellv}
\end{eqnarray}
In particular, the equality~\eqref{core} = \eqref{ellv} shows that at equal times we can commute the time-derivative and the derivative with respect to the perturbation.

\section{Conclusion}
The Donsker-Varadhan functional is related to the dynamical activity as recently studied in constructions of nonequilibrium statistical mechanics. We have given a necessary condition for its monotonicity under the Master equation. The condition was called ``normal linear response,'' as it requires the generalized susceptibility to
initially decay in time.  We have given also examples where that condition fails and where the dynamical activity starts out being non-monotone. It remains open to understand why
in those models where that sufficient condition is violated the large-time behavior of the activity is still monotone as observed numerically.

\begin{acknowledgments}
C.M.\ benefits from the Belgian Interuniversity Attraction Poles Programme P6/02. K.N.\ acknowledges the
support from the Academy of Sciences of the Czech Republic under Project No.~AV0Z10100520.
This work was initiated while B.W. completed a post-doctoral stay at the Institut de Physique Th\'eorique, CEA-Saclay, France.
\end{acknowledgments}

\end{document}